\documentclass[12pt]{article}
\pdfoutput=1

\usepackage{graphicx}
\usepackage[round]{natbib}
\usepackage{amsmath}
\usepackage[hidelinks]{hyperref}
\usepackage{authblk}
\usepackage{amsthm}

\DeclareMathOperator*{\argmin}{arg\,min}

\newtheorem{theorem}{\bf Theorem}[section]
\newtheorem{proposition}[theorem]{Proposition}

\begin{document}

\title{Unit-free and robust detection of differential expression from RNA-Seq data}

\author[1,2,*]{Hui Jiang} \author[1]{Tianyu Zhan}
\affil[1]{Department of Biostatistics, University of Michigan}
\affil[2]{Center for Computational Medicine and Bioinformatics, University of Michigan}
\affil[*]{Please send correspondence to jianghui@umich.edu.}
\maketitle

\begin{abstract}
Ultra high-throughput sequencing of transcriptomes (RNA-Seq) is a widely used method for quantifying gene expression levels due to its low cost, high accuracy and wide dynamic range for detection. However, the nature of RNA-Seq makes it nearly impossible to provide absolute measurements of transcript abundances. Several units or data summarization methods for transcript quantification have been proposed in the past to account for differences in transcript lengths and sequencing depths across different genes and different samples. Nevertheless, further between-sample normalization is still needed for reliable detection of differentially expressed genes. In this paper we propose a unified statistical model for joint detection of differential gene expression and between-sample normalization. Our method is independent of the unit in which gene expression levels are summarized. We also introduce an efficient algorithm for model fitting. Due to the L0-penalized likelihood used in our model, it is able to reliably normalize the data and detect differential gene expression in some cases when more than $50\%$ of the genes are differentially expressed in an asymmetric manner. We compare our method with existing methods using simulated and real data sets.

\end{abstract}

\section{Introduction}
\label{sec:intro}

Ultra high-throughput sequencing of transcriptomes (RNA-Seq) is a widely used method for quantifying gene expression levels due to its low cost, high accuracy and wide dynamic range for detection~\citep{Mortazavi2008}. As of today, modern ultra high-throughput sequencing platforms can generate hundreds of millions of sequencing reads from each biological sample in a single day. RNA-Seq also facilitates the detection of novel transcripts~\citep{Trapnell2010} and the quantification of transcripts on isoform level~\citep{Jiang2009,Salzman2011}. For these reasons, RNA-Seq has become the method of choice for assaying transcriptomes~\citep{Wang2009}. 

In a typical RNA-Seq experiment, mRNA transcripts are extracted from biological samples, reverse transcribed into cDNA molecules, randomly fragmented into pieces, filtered based on fragment lengths (size selection), linked with sequencing adapters and finally processed by a sequencer. The data output by the sequencer are sequenced reads (or reads for short), from either one end (i.e., single-end sequencing) or both ends (i.e., paired-end sequencing) of the fragments. These reads are usually aligned to reference transcripts or genomes and the data are then summarized as read counts for each transcript in each sample. Complications may occur when reads cannot be uniquely aligned to reference transcripts or genomes, either due to sequence homology among genes, or due to multiple transcripts (isoforms) sharing commons regions (exons) within a single gene. In this paper, we will ignore these complications and assume that each gene has only one transcript and therefore we will use terms ``gene'' and ``transcript'' interchangeably. We will also assume that each read can be uniquely aligned to a single gene, and we take read counts for each gene in each sample as our data. However, our approach can work with estimated read counts or gene expression levels from methods developed to handle these complications such as~\citet{Jiang2009} and~\citet{Li2010a}. An overview of these methods is given in~\citet{Pachter2011}.

One major limitation of RNA-Seq is that it only provides relative measurements of transcript abundances. Because reads are sequenced from a random sample of transcript fragments, changing the total amount of transcripts in a sample will have little effect on the distribution of sequenced reads. Furthermore, longer transcripts will generate more fragments which will subsequently result in more reads, and sequencing with higher depth will also lead to more reads for each transcript. To account for these factors, several data summarization units (a.k.a. within-sample normalization methods) for transcript quantification have been proposed in the past to account for different sequencing depths across samples and possibly also for different transcript lengths across genes, which include CPM/RPM (counts/reads per million)~\citep{Robinson2010a}, RPKM/FPKM (reads/fragments per kilobase of exon per million mapped reads)~\citep{Mortazavi2008,Trapnell2010} and TPM (transcript per million)~\citep{Li2010a}. Since a ``read'' can refer to either a single-end read or a paired-end read, depending on the sequencing experiment conducted, we will use the term RPKM to represent both RPKM and FPKM in this paper. Suppose that there are a total of $m$ genes in the sample. For $i=1,\ldots,m$, let $l_i$ be the length (often measured as the effective length after adjusted for edge effects) of gene $i$ and let $c_i$ be the observed read count for gene $i$ in the sample. CPM (denoted as $cpm_i$), RPKM (denoted as $rpkm_i$) and TPM (denoted as $tpm_i$) values for gene $i$ are defined as follows, respectively
\begin{equation}
\label{units}
\begin{array}{lll}
cpm_i&=& 10^6c_i/\sum_ic_i\\
rpkm_i&=& 10^3cpm_i/l_i\\
tpm_i&=& 10^6rpkm_i/\sum_irpkm_i
\end{array}
\end{equation}
In the past, all these units have been used (with RPKM being the most widely used one) for quantifying gene expression levels from RNA-Seq data. It has been argued that TPM should be used instead of RPKM since TPM estimates the relative molar concentration of transcripts in a sample~\citep{Wagner2012}. However, none of these methods can be used directly to detect differentially expressed (DE) genes reliably without a further between-sample normalization step, which is necessary to make gene expression measurements comparable across samples. Different between-sample normalization methods make different assumptions on the distribution of gene expression levels across samples. For instance, quantile normalization~\citep{Bolstad2003} assumes that the overall distributions of gene expression levels are the same for all the samples. In fact, both RPKM and TPM can also be considered as between-sample normalization methods when they are used directly to detect DE genes without additional normalization -- TPM assumes that the total numbers of transcripts (i.e., total molar amount) are the same for all the samples and RPKM assumes that the total numbers of nucleotides in all the transcripts (i.e., total physical mass) are the same for all the samples, both of which are strong but arguably reasonable assumptions.

Between-sample normalization also has limitations. Consider a simple hypothetical example of comparing gene expression profiles of two samples A and B, where $60\%$ of the genes were up-regulated by 2-fold in sample B while the other $40\%$ of the genes stayed constant in both samples. Due to the relative nature of RNA-Seq measurements, it is impossible to distinguish it from the scenario where the first group of $60\%$ genes actually stayed constant but the second group of $40\%$ gene were down-regulated by 2-fold in sample B. Since the problem is non-identifiable in such cases, the normalization approach has to rule out the ambiguities based on its assumptions. One commonly used assumption is that the majority (i.e., $>50\%$) of the genes are non-DE. The median ratio method~\citep{Anders2010} and TMM (trimmed mean of M values)~\citep{Robinson2010a} are two normalization methods based on this assumption. Furthermore, between-sample normalization and detection of DE gene are two problems that are always tangled together, since ideally normalization should be based on non-DE genes only. The iterative normalization approach~\citep{Li2012} utilizes this idea and iterates between normalization and detection of DE genes. However, it is unclear what objective function is optimized in such an approach, and the final solution often depends on the initial guess which is undesirable. For an overview of between-sample normalization approaches and comparison of their performance for detection of DE genes, please refer to~\citet{Dillies2013,Rapaport2013}.

Fortunately, the example described above is rather extreme and unrealistic, because in practice genes rarely change at the same pace -- it is rather unlikely that all the DE genes are up-regulated by exactly the same amount (2-fold in the above example). Therefore, a more practical assumption is that among DE genes, the degrees at which genes change have an unknown but spreaded distribution. This assumption, while largely ignored by most existing approaches, will be exploited in this paper. 

In this paper, we will propose a unified statistical model for joint detection of differential gene expression and between-sample normalization. We will introduce the model and  an efficient algorithm for model fitting in Section~\ref{sec:model}. Comparisons with existing methods in simulated and real data sets will be given in Section~\ref{sec:experiments}, followed by discussions in Section~\ref{sec:discussion}.

\section{A penalized likelihood approach}
\label{sec:model}

\subsection{The model}\label{subsec:model}

Suppose there are a total of $m$ genes measured in $S$ groups of experiments with $n_1,\ldots,n_S$ samples, respectively. Let $x_{sij}, s=1,\ldots,S, i=1,\ldots,m, j=1,\ldots,n_s$ be log-transformed gene expression measurement for the $i$-th gene in the $j$-th sample in the $s$-th group. Here, $x_{sij}$ can be gene expression measurement summarized in units such as log(count), log(CPM), log(RPKM) or log(TPM), while a small positive number is usually added to the raw read count before any calculation to avoid taking logarithm of zero. The following statistical model is assumed
\begin{equation}
\label{model}
	x_{sij} \sim N(\mu_{si} + d_{sj}, \sigma^2_{i})
\end{equation}
where $\mu_{si}$ is the mean of log-transformed expression levels of gene $i$ in group $s$, $d_{sj}$ is a scaling factor (e.g., $\log(\mbox{sequencing depth})$ or $\log(\mbox{library size})$) for sample $j$ in group $s$ and $\sigma^2_{i}$ is the variance of log-transformed expression levels of gene $i$ across all $S$ groups. Here we assume that the log-transformation stabilizes the variances and makes them roughly the same across groups (yet can still be different across genes). Nevertheless, our model can be extended to accommodate heteroscedastic variances across groups.

Our main interest is in detecting differentially expressed (DE) genes across the $S$ groups. Let $\tau_i$ be the indicator of differential expression for gene $i$ such that $\tau_i=1$ if gene $i$ is differentially expressed across the $S$ groups and $\tau_i=0$ otherwise, i.e., $\tau_i=1$ if and only if $\mu_{1i}=\mu_{2i}=\cdots=\mu_{Si}$. The parameters of major interest are $\{\tau_i\}_{i=1}^m$, while $\mu_{si}, d_{sj}$ and $\sigma^2_i$ might be of interest too, because they denote biologically meaningful quantities.

Since there are many data summarization units (i.e., count, CPM, RPKM or TPM) for RNA-Seq data, it is desirable for a method for detecting DE genes to be independent of the unit in which gene expression levels are summarized. That is, the inference result (i.e., $\tau_i$) remains the same when the data (i.e., $x_{sij}$) change from one unit (e.g., log(count)) to another (e.g., log(RPKM)). Due to the structure of our model, as well as the log-transformation of the data, it can be shown that our model enjoys this property (see Appendix for proof). To the best of our knowledge, no other existing model for DE detection has this property.

\subsection{Penalized likelihood}

For now, we assume that $\{\sigma^2_i\}_{i=1}^m$ are known. In practice, we solve for $\sigma^2_i$ using an iterative approach, which will be described later.

To fit model~(\ref{model}), we reparametrize $\mu_{si}$ as $\mu_i=\mu_{1i}, \gamma_{si}=\mu_{si}-\mu_{1i}, s=2,\ldots,S$. Then the model becomes
\begin{equation}
\label{model2}
\left\{
\begin{array}{l}
	x_{1ij} \sim N(\mu_i + d_{1j}, \sigma_i^2) \\
	x_{sij} \sim N(\mu_i + \gamma_{si} + d_{sj}, \sigma_i^2), s=2,\ldots,S \\
\end{array}
\right.
\end{equation}	
To fit model~(\ref{model2}), we minimize its negative log-likelihood
$$l(\mu, \gamma, d; x)=\sum_{i=1}^m\frac1{2\sigma_i^2}\left(\sum_{j=1}^{n_1}(x_{1ij}-\mu_i-d_{1j})^2+\sum_{s=2}^S\sum_{j=1}^{n_s}(x_{sij}-\mu_i-\gamma_{si}-d_{sj})^2\right)$$
where the term $(n/2)\sum_{i=1}^m\log(2\pi\sigma^2_i)$ is discarded as it does not contain any unknown parameter when $\sigma^2_i$ is known. Model~($\ref{model2}$) is non-identifiable because we can simply add any constant to all the $d_{sj}$'s and subtract the same constant from all the $\mu_i$'s, while having the same fit for $l(\cdot)$. To resolve this issue, we fix $d_{11}=0$. Furthermore, we introduce a sparse penalty $p(\gamma)$ on all the $\gamma_{si}$'s and formulate a penalized likelihood
\begin{equation}
\label{penalized_likelihood}
f(\mu,\gamma,d)=l(\mu, \gamma, d; x)+p(\gamma)
\end{equation}
Commonly used sparse penalty functions for $p(\gamma)$ are L1 (a.k.a. lasso)~\citep{Tibshirani1996}, L0, SCAD~\citep{Fan2001} and etc. The sparse penalty will force some of the $\gamma_{si}$'s to become exactly zero, which will in turn facilitate the detection of DE genes since by definition $\tau_i=1(\sum_{s=2}^S|\gamma_{si}|>0)$.

Let $n=\sum_{s=1}^Sn_s$ be the total sample size. There are $mn$ observations and $mS+n-1$ free parameters in the model, and typically we can have $n$ in tens or hundreds and $m$ in tens of thousands. Using an L1 penalty, it will be computationally intensive if we fit the model using a lasso solver such as Glmnet~\citep{Friedman2010}, and typically it will be computationally even more challenging to fit the model with a non-convex penalty such as L0 or SCAD. Fortunately, we can take advantage of the structure in model~(\ref{model2}) and solve it efficiently. In this paper we work with the L0 penalty due to its robustness in estimation and variable selection
\begin{equation}
\label{L0}
p(\gamma)=\sum_{i=1}^m\alpha_i1(\sum_{s=2}^S|\gamma_{si}|>0)
\end{equation}
where $\alpha_i>0, i=1,\ldots,m$ are tuning parameters. The approach to choose $\{\alpha_i\}_{i=1}^m$ will be described later. Our model fitting approach can also be adapted to accommodate other penalty functions. 

\subsection{Model fitting}

It can be shown that the solution to~(\ref{penalized_likelihood}) with penalty~(\ref{L0}) can be obtained as follows (see Appendix for proof). 
\begin{proposition}
\label{solution}
Model~(\ref{penalized_likelihood}) with penalty~(\ref{L0}) can be solved as follows
$$
\begin{array}{l}
	d^\prime_{sj} = (\sum_{i=1}^m(x_{sij}-x_{si1})/\sigma_i^2)/(\sum_{i=1}^{m}1/\sigma_i^2), s=1,\ldots,S\\
	\mu^\prime_{si} = (1/n_s)\sum_{j=1}^{n_s}(x_{sij}-d^\prime_{sj}) , s=1,\ldots,S\\
	%\lambda=\sqrt{\frac{(n_1+n_2)\alpha}{n_1n_2}}\\
	d_1=0\\
        d_2,\ldots,d_S=\displaystyle\argmin_{d_2,\ldots,d_S}\sum_{i=1}^m\min\left(g(d_2,\ldots,d_S), \alpha_i\right)\\
        \mbox{where } \displaystyle g(d_2,\ldots,d_S)=\frac1{2\sigma_i^2}\left\{\sum_{s=1}^Sn_s(\mu_{si}^\prime-d_s)^2-\frac1n\left[\sum_{s=1}^S(n_s(\mu_{si}^\prime-d_s))\right]^2\right\}\\
	d_{sj}=d_s+d^\prime_{sj} , s=1,\ldots,S\\
	$$
	\gamma_{si} = \left\{
	\begin{array}{ll}
		0  &\mbox{ if } g(d_2,\ldots,d_S)<\alpha_i\\
		\mu^\prime_{si}-\mu^\prime_{1i}-d_s &\mbox{ otherswise}		
	\end{array}\right.
	$$ \\
	\mu_i = \left\{
	\begin{array}{ll}
		$$(1/n)\sum_{s=1}^Sn_s(\mu^\prime_{si}-d_s)$$ &\mbox{ if }  g(d_2,\ldots,d_S)<\alpha_i\\
		\mu^\prime_{1i} &\mbox{ otherwise}\\
	\end{array}\right.\\
\end{array}
$$
%\sigma_{si}^2=\frac1{n_s-1}\sigma_{si,\: LOWESS}^2+\frac{n_s-2}{n_s-1}\sigma_{si,\:gene\:wise}^2 \\
%\sigma_i^2=\frac{\sum_{s=1}^S(n_s-1)\sigma_{si}^2}{\sum_{s=1}^Sn_s-S}%\sum_{s=2}^S|\gamma_{si}|=0
\end{proposition}
In Proposition~\ref{solution}, the only computationally intensive step is to solve for $d_2,\ldots,d_S$, for which the following function is minimized
$$G(d_2,\ldots,d_S)=\sum_{i=1}^m\min\left(g(d_2,\ldots,d_S), \alpha_i\right).
$$
Typically, function $G(\cdot)$ is non-convex and non-differentiable. Such kind of functions are usually difficult to optimize. However, in low-dimensional cases, function $G(\cdot)$ can be minimized efficiently using exhaustive search. Examples with $S=2$ and $S=3$ are given in Figure~\ref{two-min-d} (where $G^\prime(\cdot)$, a variation of $G(\cdot)$ is shown; see Section~\ref{subsec:twogroup} for details) and Figure~\ref{three-min-d}, respectively.

\begin{figure}[htb]
\begin{center}
\includegraphics[scale=0.5]{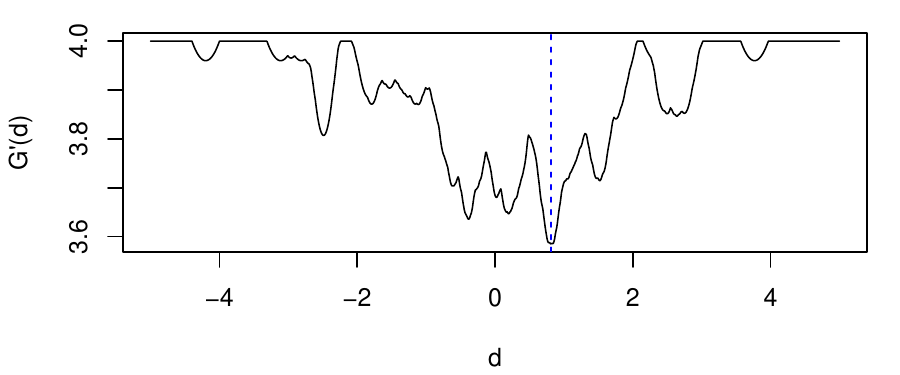}\\
\footnotesize{(a)}\\
\includegraphics[scale=0.5]{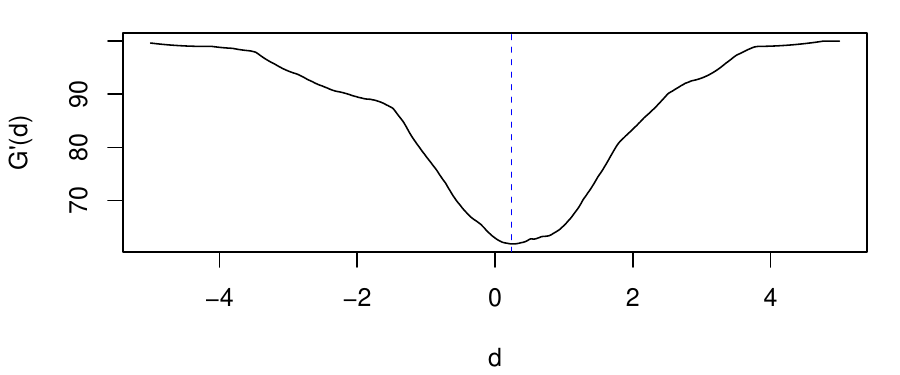}\\
\footnotesize{(b)}\\
\end{center}
\caption{Function $G^\prime(d)=\sum_{i=1}^{100}(1/\sigma^2_i)\min\left((\mu^\prime_{2i}-\mu^\prime_{1i}-d)^2, \lambda_i^2\right)$ with simulated $\{\mu^\prime_{1i}\}_{i=1}^{100}, \{\mu^\prime_{2i}\}_{i=1}^{100}\sim N(0,1), \{\sigma^2_i\}_{i=1}^{100}=1$ with (a) $\{\lambda_i\}_{i=1}^{100}=0.2$ and (b) $\{\lambda_i\}_{i=1}^{100}=1$. The minimizers of $G^\prime(\cdot)$ are shown with dashed lines. \label{two-min-d}}
\end{figure}

\begin{figure}[htb]
\begin{center}
\begin{tabular}{cc}
\includegraphics[scale=0.4]{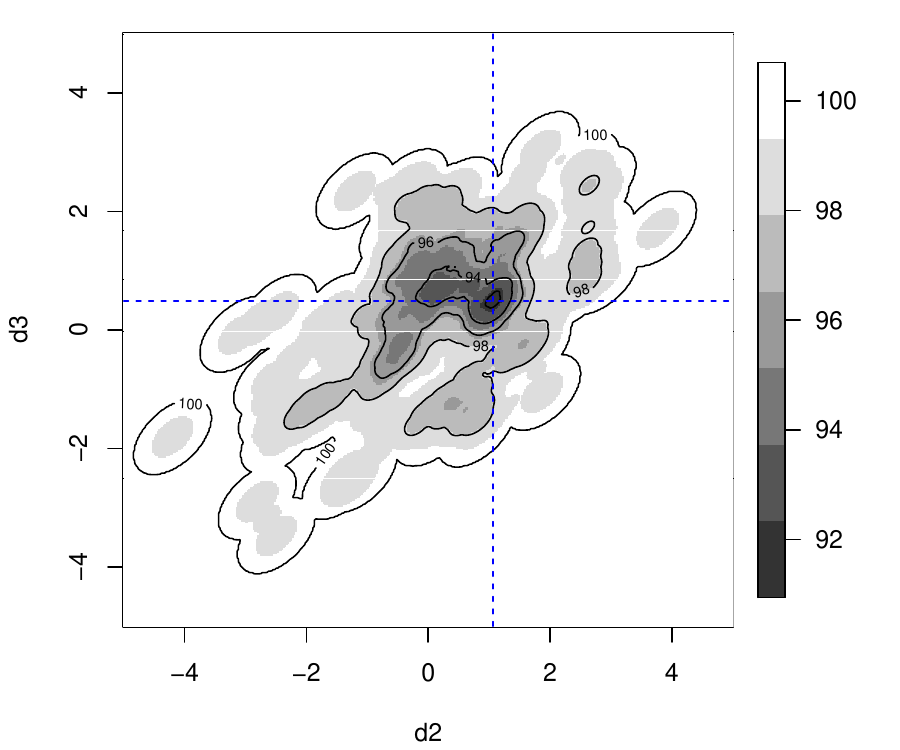} & \includegraphics[scale=0.4]{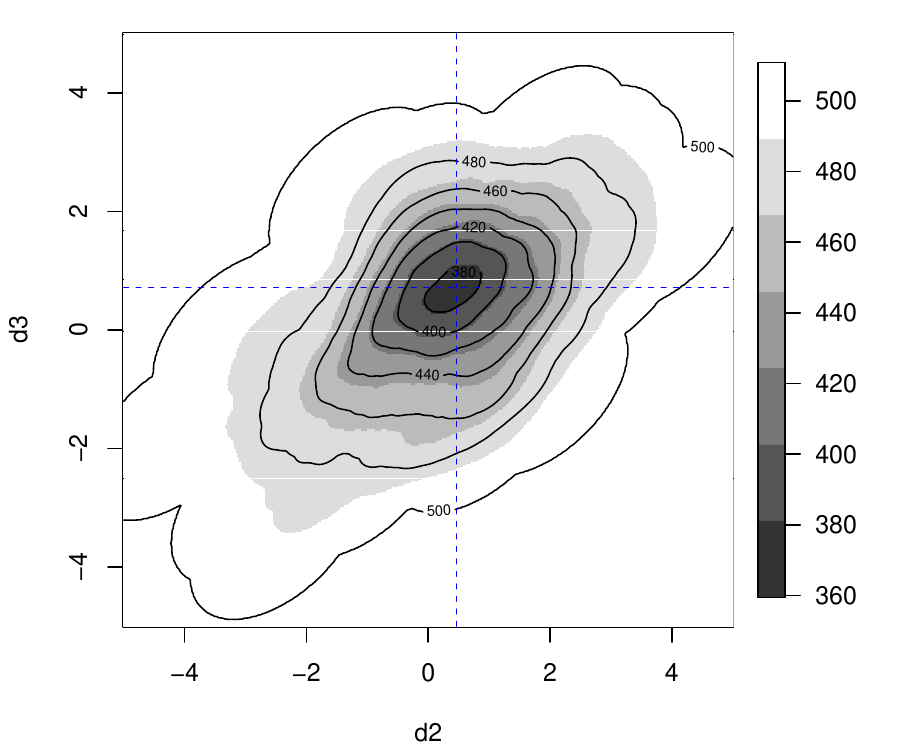} \\ 
\footnotesize{(a)} & \footnotesize{(b)} \\ 
\end{tabular} 
\end{center}
\caption{Function $G(d2,d3)$ with simulated $\{\mu^\prime_{1i}\}_{i=1}^{100}, \{\mu^\prime_{2i}\}_{i=1}^{100},\{\mu^\prime_{3i}\}_{i=1}^{100}\sim N(0,1), \{\sigma^2_i\}_{i=1}^{100}=1, n_1=n_2=n_3=10$ with (a) $\{\alpha_i\}_{i=1}^{100}=1$ and (b) $\{\alpha_i\}_{i=1}^{100}=5$. The minimizers of $G(\cdot)$ are at the intersections of the dashed lines. \label{three-min-d}}
\end{figure}

\subsection{Choosing $\alpha_i$}
\label{subsec:alpha}

Choosing the tuning parameters in a penalized regression model is usually quite challenging because it involves the bias-variance tradeoff~\citep{Hastie2009}. In practice, cross-validation can often achieve acceptable performance, at the cost of additional computation and less robustness. Fortunately, there is a simple way to choose the tuning parameters in our model, which is based on the property of the solution in Proposition~\ref{solution}. Let $y_{sij}=x_{sij}-d_{sj}$ denote the normalized data. Then 
$$\mu_{si}^\prime-d_s=\frac1{n_s} \sum_{j=1}^{n_s}(x_{sij}-d^\prime_{sj}-d_s)=\frac1{n_s}\sum_{j=1}^{n_s}y_{sij}$$
is the mean of $y_{sij}$ for gene $i$ in group $s$. The condition for $\gamma_{si}=0$ in Proposition~\ref{solution} can be rewritten as
$$
\frac1{(S-1)\sigma_i^2}\left\{\sum_{s=1}^S\frac{\left(\sum_{j=1}^{n_s}y_{sij}\right)^2}{n_s}-\frac1n\left(\sum_{s=1}^S\sum_{j=1}^{n_s}y_{sij}\right)^2\right\}<\frac{2\alpha_i}{S-1}\\
$$
where the left hand side (LHS) has the form of the $F$-statistic for one-way ANOVA models, which suggests we choose $\alpha_i$ as $((S-1)/2)F^*_{1-q}(S-1, n-S)$, where $F^*_{1-q}(\cdot)$ is the critical value for one-sided level $q$ tests with the $F$-distribution. Typical $p$-value cutoffs can be used for $q$ here, such as $0.05$, $0.01$ or even lower values for more stringent DE gene detection. We set $q=0.01$ in our experiments.

\subsection{Simplification for two-group comparison}
\label{subsec:twogroup}

Here we study the simplest case of two-group comparison, which is the most widely used experimental design for differential gene expression study. For $S=2$, the condition for $\gamma_{2i}=0$ is 
$$\frac1{2\sigma_i^2}\frac{n_1n_2}n(\mu_{2i}^\prime-\mu_{1i}^\prime-d_2)^2<\alpha_i$$
which can be rewritten as
$$|\mu^\prime_{2i}-\mu^\prime_{1i}-d| < \lambda_i$$
where 
$$\lambda_i=\sqrt{\frac{2n\sigma^2_i\alpha_i}{n_1n_2}}$$
is a tuning parameter that is alternative to and more convenient than $\alpha_i$. Therefore, the solutions for $d_2$ (denoted as $d$ for simplicity), $\gamma_{2i}$ (denoted as $\gamma_i$ for simplicity) and $\mu_i$ in Proposition~\ref{solution} can be simplified as
$$
d=\argmin_{d}G^\prime(d)=\argmin_{d}\sum_{i=1}^m\frac1{\sigma_i^2}\min\left((\mu^\prime_{2i}-\mu^\prime_{1i}-d)^2, \lambda_i^2\right) 
$$
$$
\gamma_i = \left\{
	\begin{array}{ll}
		0  &\mbox{ if } |\mu^\prime_{2i}-\mu^\prime_{1i}-d| < \lambda_i\\
		\mu^\prime_{2i}-\mu^\prime_{1i}-d &\mbox{ otherswise}		
	\end{array}\right. 
$$
$$
\mu_i = \left\{
	\begin{array}{ll}
		(n_1\mu^\prime_{1i}+n_2(\mu^\prime_{2i}-d))/n &\mbox{ if } |\mu^\prime_{2i}-\mu^\prime_{1i}-d| < \lambda_i\\
		\mu^\prime_{1i} &\mbox{othersise}\\
	\end{array}\right.
$$
Similarly, the condition for $\gamma_i=0$ can be rewritten as
$$
\frac{\left|(1/n_2)\sum_{j=1}^{n_2}y_{2ij}-(1/n_1)\sum_{j=1}^{n_1}y_{1ij}\right|}{\sigma_i\sqrt{1/n_1+1/n_2}}<\frac{\lambda_i}{\sigma_i\sqrt{1/n_1+1/n_2}}
$$
where the LHS has the form of the $t$-statistic for two-group comparison with pooled variance, which suggests we choose $\lambda_i$ as $t^*_{1-\frac{q}2}(n_1+n_2-2)\sigma_i\sqrt{1/n_1+1/n_2}$, where $t^*_{1-\frac{q}2}(\cdot)$ is the critical value for two-sided level $q$ tests with the $t$-distribution. The corresponding value for $\alpha_i$ is $(1/2)t^*_{1-\frac{q}2}(n_1+n_2-2)^2$ which is the same as $((S-1)/2)F^*_{1-q}(S-1, n-S)$ suggested in Section~\ref{subsec:alpha}.

\subsection{Solving for $\{\sigma^2_i\}_{i=1}^m$}
\label{sec:sigma2}

To solve for $\{\sigma^2_i\}_{i=1}^m$, consider the negative log-likelihood function (with $\{\sigma^2_i\}_{i=1}^m$ being unknown parameters as well) for model~(\ref{model}) restricted to group $s$ 
$$l_s(\mu, d, \sigma^2; x)=\sum_{i=1}^m\left(\frac{n_s}2\log(2\pi\sigma^2_i)+\frac1{2\sigma^2_i}\sum_{j=1}^{n_s}(x_{sij}-\mu_{si}-d_{sj})^2\right)$$
Taking partial derivatives of $l_s(\cdot)$ with respect to $\mu_{si}, d_{sj}$ and $\sigma^2_i$ respectively and setting the partial derivatives to be zero, we have
$$\mu_{si}=\frac1{n_s}\sum_{j=1}^{n_s}(x_{sij}-d_{sj})$$
$$d_{sj}=\frac{\sum_{i=1}^{m}\frac1{\sigma^2_i}(x_{sij}-\mu_{si})}{\sum_{i=1}^{m}\frac1{\sigma^2_i}}$$
$$\sigma^2_i=\frac1{n_s}\sum_{j=1}^{n_s}(x_{sij}-\mu_{si}-d_{sj})^2$$
Estimates for $\mu_{si}, d_{sj}$ and $\sigma^2_i$ can then be iteratively updated using the above three equations until converge, which is similar to the method of iteratively reweighted least squares for two-way ANOVA models with heteroscedastic errors. In our implementation, we use the following modified equation for updating $\sigma^2_i$ to reduce the estimation bias.
$$\sigma^2_i=\frac1{n_s-1}\sum_{j=1}^{n_s}(x_{sij}-\mu_{si}-d_{sj})^2$$
Same as before, we fix $d_{s1}=0$ and adjust the remaining $d_{sj}$'s and $\mu_{si}$'s accordingly after each iteration to resolve the non-identifiability issue. Note that $\mu_{si}$ and $d_{sj}$ estimated here are based on group $s$ data only and therefore will be discarded. They should not be confused with the parameters estimated from model~(\ref{model}) based on the data from all the groups.

Denote the variance estimated using the above iterative algorithm for gene $i$ in group $s$ as $s^2_{si}$. We take a weighted average of $\{s^2_{si}\}_{s=1}^S$ to pool information from all the groups
$$s^2_i = \frac{\sum_{s=1}^S(n_s-1)s^2_{si}}{n-S}$$
Then we take another weighted average of $s^2_{i}$ and the estimated mean variance across all the genes to obtain a robust estimate for $\sigma^2_i$. That is
$$\widehat{\sigma^2_{i}}=(1-w)s^2_{i}+w\overline{s^2}$$
where $\overline{s^2}=\sum_{i=1}^ms^2_{i}/m$, and the weight $w$ is calculated using the following formula as suggested in~\citet{Ji2005} which is based on an empirical Bayes approach
$$w=\frac{2(m-1)}{n-S+2}\left(\frac1m+\frac{(\overline{s^2})^2}{\sum_{i=1}^m(s^2_i-\overline{s^2})^2}\right)$$
This kind of variance estimation approach is widely used in differential gene expression analysis with small sample sizes~\citep{Ji2010,Smyth2004}. Using the estimated variances $\widehat{\sigma^2_{i}}$, we can then solve for $d_s$, $\gamma_{si}$ and $\mu_i$ as described in Proposition~\ref{solution}.

\subsection{Ranking the genes}
To facilitate downstream analyses, after model fitting, we can calculate a $p$-value for each gene using the $F$-test on the normalized data $y_{sij}=x_{sij}-d_{sj}$ with the estimated variance $\widehat{\sigma_i^2}$. The test statistics for the $i$-th gene is
$$
F_i=\frac1{(S-1)\widehat{\sigma_i^2}}\left\{\sum_{s=1}^S\frac{\left(\sum_{j=1}^{n_s}y_{sij}\right)^2}{n_s}-\frac1n\left(\sum_{s=1}^S\sum_{j=1}^{n_s}y_{sij}\right)^2\right\}$$
$$=\frac1{(S-1)\widehat{\sigma_i^2}}\left\{\sum_{s=1}^Sn_s(\mu_{si}^\prime-d_s)^2-\frac1n\left[\sum_{s=1}^S(n_s(\mu_{si}^\prime-d_s))\right]^2\right\}\\
$$
and the degrees of freedom (DF) of the reference distribution are $(S-1, n-S)$.

When comparing two groups (e.g., group 2 vs group 1), it is equivalent to the two-sample $t$-test with test-statistic
$$t_{i}=\frac{\left|(1/n_2)\sum_{j=1}^{n_2}y_{2ij}-(1/n_1)\sum_{j=1}^{n_1}y_{1ij}\right|}{\sqrt{\widehat{\sigma_i^2}(1/n_1+1/n_2)}}=\frac{(\mu_{2i}^\prime-d_2)-(\mu_{1i}^\prime-d_1)}{\sqrt{\widehat{\sigma_i^2}(1/n_1+1/n_2)}}$$
and DF of $n_1+n_2-2$.

The $p$-values resulted from these tests or the corresponding false discovery rate (FDR) values (e.g., estimated using the Benjamini-Hochberg (BH) procedure~\citep{benjamini1995controlling}) can then be used to rank the genes to facilitate downstream analysis (e.g., the receiver operating characteristic (ROC) curve analysis as in our experiments) or to prioritize follow-up experiments. 

Similarly, log-fold change values between two groups (e.g., group 2 vs group 1) can be estimated as
$$logFC_i=(\mu_{2i}^\prime-d_2)-(\mu_{1i}^\prime-d_1).$$

\section{Experiments}
\label{sec:experiments}

\subsection{Simulation of two-group data}
We simulate RNA-Seq data with a total of $m=1000$ genes ($700$ non-DE genes, $300$ DE genes) in two-group as follows
$$\begin{array}{ll}
	\mu_{1i}=\mu_{2i}\sim N(-3,2) &\mbox{ mean log expression for genes 1-700}\\
	\mu_{1i}\sim N(-3,2), \mu_{2i}=\mu_{1i}+ N(0,1) &\mbox{ mean log expression for genes 701-1000}\\
	d_{sj}\sim N(0, 0.5) &\mbox { log scaling factor} \\
	x_{sij}\sim N(\mu_{si}+d_{sj}, 0.2) &\mbox{ log gene expression}\\
	\log(l_i)\sim Unif(5,10) &\mbox{ log gene lengths} \\
	N_{sj} \sim Unif(3, 5)\times10^7 &\mbox{ library sizes} \\
	\displaystyle\{c_{sij}\}_{i=1}^m \sim Mult\left(N_{sj},\frac{{l_ie^{x_{sij}}}_{i=1}^m}{\sum_{i=1}^ml_ie^{x_{sij}}}\right)+1 &\mbox{ read counts}\\
\end{array}
$$
We use $\log{c_{sij}}$ as data to fit our model. The fitted $\{\gamma_i\}_{i=1}^m$ are plotted in Figure~\ref{gamma}a. Using CPM, RPKM or TPM values computed with formulas in~(\ref{units}) yield the same result. To demonstrate the robustness of our method, we also simulate with $300$ non-DE genes and $700$ DE genes. Furthermore, for DE genes we simulate with $\mu_{2i}=\mu_{1i}+N(1,1)$, which means that the average log-fold change is $1$. Our method still robustly estimates the $\gamma_i$'s (Figure~\ref{gamma}b). We further simulate with $100$ non-DE genes and $900$ DE genes, for which our method still achieves robust estimates when we simulate with $\mu_{2i}=\mu_{1i}+N(1,1)$ (Figure~\ref{gamma}c). Only when we simulate with $900$ DE genes and with $\mu_{2i}=\mu_{1i}+N(3,1)$, our method fails to achieve robust estimates (Figure~\ref{gamma}d).

\begin{figure}[htb]
\begin{center}
\includegraphics[scale=0.4]{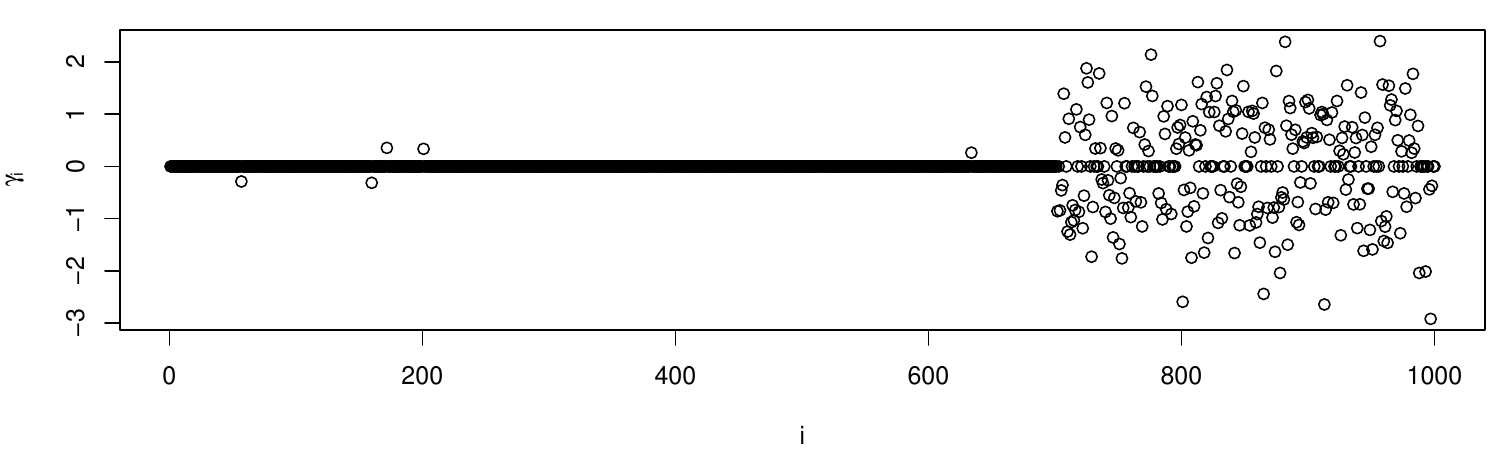}\\
\footnotesize{(a)}\\
\includegraphics[scale=0.4]{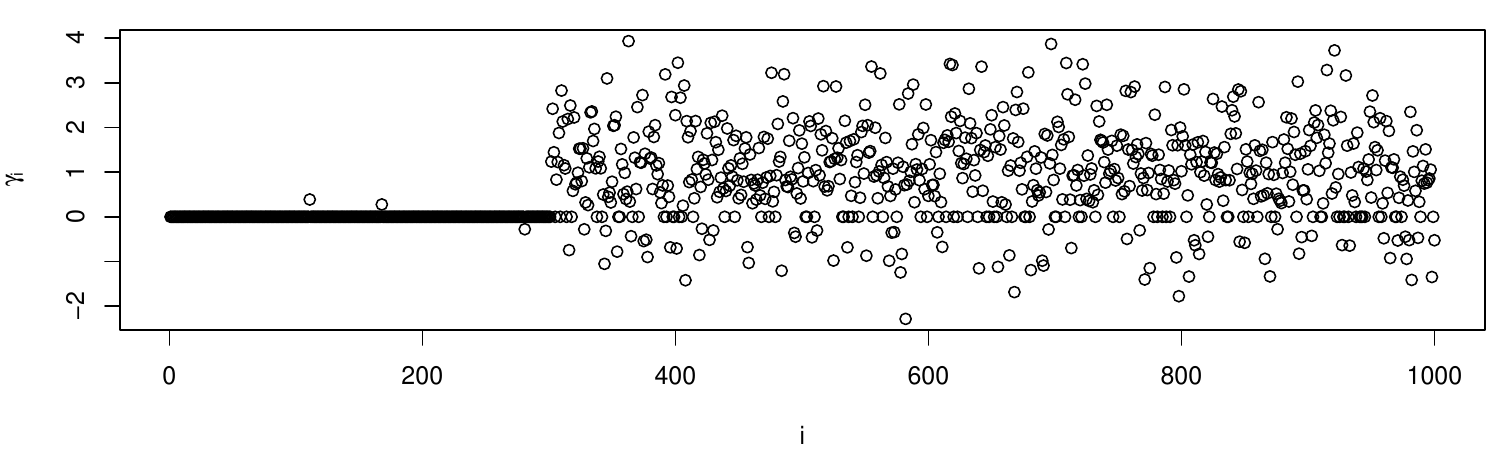}\\
\footnotesize{(b)}\\
\includegraphics[scale=0.4]{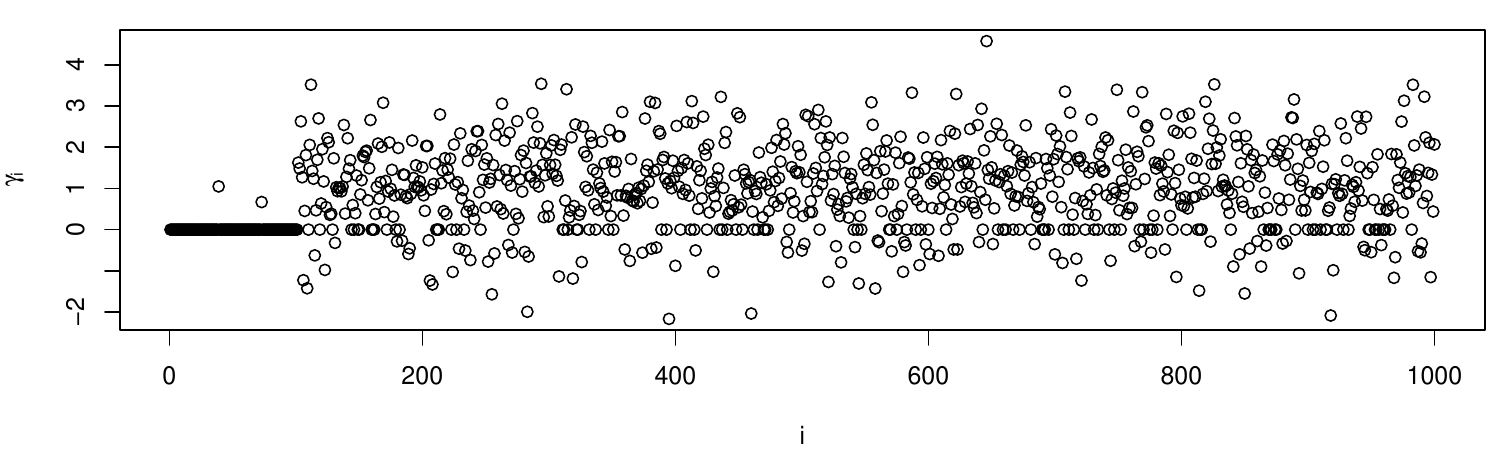}\\
\footnotesize{(c)}\\
\includegraphics[scale=0.4]{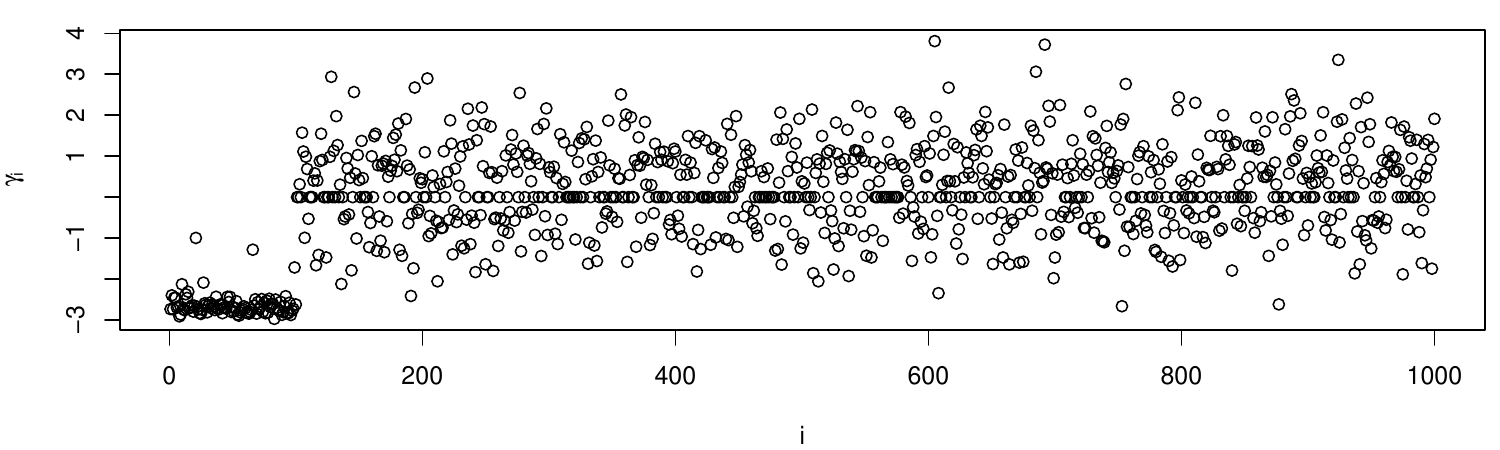}\\
\footnotesize{(d)}\\
\end{center}
\caption{Estimated $\gamma$ from simulated two-group data. \label{gamma}}
\end{figure}

\subsection{Comparison with existing methods using simulated RNA-Seq data}
We compare our method (named rSeqRobust) with edgeR-robust~\citep{edgeR,edgeR-robust}, DESeq2~\citep{DESeq2}, and limma-voom~\citep{limma,limma-voom}, all of which are state-of-the-art methods for detecting differential gene expression from RNA-Seq data. We simulate two-group RNA-Seq data with a total of $m$ = 1000 genes, and different sample sizes ranging from small sample size (two samples in each group, i.e., $n_1=n_2=2$) to large sample size (twelve samples in each group). We simulate both log-normally (LN) distributed read counts, which is the model assumptions of limma-voom and rSeqRobust, as well as negative-binomially (NB) distributed read counts, which is the underlying assumption of edgeR-robust and DESeq2. The distributions of gene expression levels and the distributions of library sizes for both simulations, as well as the distributions of read count dispersions for the NB simulation, are based on a real RNA-Seq dataset~\citep{Pickrell2010}. In the LN simulation, log read counts are assumed to be normally distributed with $\sigma=0.5$.

The simulations are performed using the simulator described in~\citep{edgeR-robust}. In its implementation, edgeR-robust and limma-voom use the TMM (trimmed mean of M values) normalization method proposed in~\citep{Robinson2010a}, and DESeq2 uses the median ratio normalization method proposed in~\citep{Anders2010}. We slightly modify the simulator to allow LN distributed data, as well as variable fold changes. We simulate data sets with 30\% or 70\% DE genes, as well as 50\%, 70\% or 90\% up-regulated genes among all the DE genes. The log-fold change for DE genes (when measured as up-regulation from one group to the other) are assumed to be distributed as $N(\log3, 1)$. Following~\citep{edgeR-robust}, we rank the genes according to the p-values reported by each of the four methods and then perform the ROC curve analysis based on the gene rankings and use the area under the curve (AUC) to evaluate the performance of the four methods for DE gene detection. The results for LN distributed data as well as for NB distributed data are summarized in Tables~\ref{tab:LN} and~\ref{tab:NB}, respectively. For each parameter setting, the highest AUC value is shown in bold font. We can see that rSeqRobust and limma-voom are the best in LN simulations, and rSeqRobust and edgeR-robust are the best in NB simulations, which is consistent with our expectation. rSeqRobust is the best among the four methods when the sample size in each group is greater than or equal to four, even when the data are simulated using the NB distribution. In particular, when rSeqRobust is outperformed by other methods, it is usually only by a small margin and only in relatively easier cases. In difficult cases (e.g., those with $n\geq8, \mbox{DE}\%=70\%, \mbox{Up}\%=90\%$), rSeqRobust outperform other methods by a relatively larger margin.

\begin{table}[htb]
\caption{Comparison of edgeR-robust, DESeq2, limma-voom and rSeqRobust using log-normally distributed data. The table shows the total sample size ($n$), percent of DE genes (DE\%), percent of up-regulated genes among all the DE genes (Up\%), as well as the mean AUCs for all four methods measured using 10 simulated replicates. The standard errors of the mean AUCs are given in parentheses. \label{tab:LN}}
\centering
\begin{tabular}{ccccccc}
  \hline
$n$ & DE\% & Up\% & edgeR-robust & DESeq2 & limma-voom & rSeqRobust \\ 
  \hline
4 & 30 & 50 & 0.819 {\tiny(0.005)} & 0.817 {\tiny(0.005)} & {\textbf{0.824 {\tiny(0.004)}}} & 0.824 {\tiny(0.005)} \\ 
  4 & 30 & 70 & 0.818 {\tiny(0.006)} & 0.815 {\tiny(0.007)} & {\textbf{0.822 {\tiny(0.006)}}} & 0.816 {\tiny(0.006)} \\ 
  4 & 30 & 90 & 0.796 {\tiny(0.005)} & 0.791 {\tiny(0.005)} & {\textbf{0.800 {\tiny(0.005)}}} & 0.779 {\tiny(0.005)} \\ 
  4 & 70 & 50 & 0.829 {\tiny(0.003)} & 0.826 {\tiny(0.003)} & {\textbf{0.834 {\tiny(0.003)}}} & 0.834 {\tiny(0.003)} \\ 
  4 & 70 & 70 & 0.787 {\tiny(0.005)} & 0.779 {\tiny(0.005)} & {\textbf{0.791 {\tiny(0.006)}}} & 0.782 {\tiny(0.006)} \\ 
  4 & 70 & 90 & {\textbf{0.670 {\tiny(0.007)}}} & 0.655 {\tiny(0.007)} & 0.670 {\tiny(0.008)} & 0.634 {\tiny(0.007)} \\ 
  6 & 30 & 50 & 0.850 {\tiny(0.004)} & 0.843 {\tiny(0.003)} & {\textbf{0.854 {\tiny(0.003)}}} & 0.854 {\tiny(0.003)} \\ 
  6 & 30 & 70 & 0.841 {\tiny(0.005)} & 0.834 {\tiny(0.004)} & 0.846 {\tiny(0.005)} & {\textbf{0.847 {\tiny(0.004)}}} \\ 
  6 & 30 & 90 & 0.822 {\tiny(0.005)} & 0.812 {\tiny(0.005)} & 0.826 {\tiny(0.005)} & {\textbf{0.828 {\tiny(0.005)}}} \\ 
  6 & 70 & 50 & 0.852 {\tiny(0.004)} & 0.846 {\tiny(0.003)} & {\textbf{0.855 {\tiny(0.004)}}} & 0.854 {\tiny(0.004)} \\ 
  6 & 70 & 70 & 0.808 {\tiny(0.005)} & 0.800 {\tiny(0.003)} & {\textbf{0.810 {\tiny(0.005)}}} & 0.801 {\tiny(0.005)} \\ 
  6 & 70 & 90 & 0.680 {\tiny(0.012)} & 0.658 {\tiny(0.011)} & {\textbf{0.680 {\tiny(0.012)}}} & 0.673 {\tiny(0.010)} \\ 
  8 & 30 & 50 & 0.877 {\tiny(0.003)} & 0.868 {\tiny(0.002)} & {\textbf{0.881 {\tiny(0.002)}}} & 0.880 {\tiny(0.002)} \\ 
  8 & 30 & 70 & 0.867 {\tiny(0.003)} & 0.858 {\tiny(0.003)} & 0.871 {\tiny(0.003)} & {\textbf{0.876 {\tiny(0.002)}}} \\ 
  8 & 30 & 90 & 0.848 {\tiny(0.004)} & 0.832 {\tiny(0.004)} & 0.850 {\tiny(0.004)} & {\textbf{0.869 {\tiny(0.003)}}} \\ 
  8 & 70 & 50 & 0.878 {\tiny(0.003)} & 0.872 {\tiny(0.003)} & {\textbf{0.882 {\tiny(0.003)}}} & 0.882 {\tiny(0.003)} \\ 
  8 & 70 & 70 & 0.815 {\tiny(0.004)} & 0.806 {\tiny(0.005)} & 0.818 {\tiny(0.003)} & {\textbf{0.841 {\tiny(0.006)}}} \\ 
  8 & 70 & 90 & 0.672 {\tiny(0.005)} & 0.651 {\tiny(0.005)} & 0.671 {\tiny(0.005)} & {\textbf{0.759 {\tiny(0.008)}}} \\ 
  12 & 30 & 50 & 0.898 {\tiny(0.005)} & 0.894 {\tiny(0.004)} & {\textbf{0.902 {\tiny(0.005)}}} & 0.902 {\tiny(0.005)} \\ 
  12 & 30 & 70 & 0.888 {\tiny(0.006)} & 0.884 {\tiny(0.005)} & 0.894 {\tiny(0.006)} & {\textbf{0.902 {\tiny(0.006)}}} \\ 
  12 & 30 & 90 & 0.862 {\tiny(0.005)} & 0.853 {\tiny(0.006)} & 0.867 {\tiny(0.006)} & {\textbf{0.899 {\tiny(0.006)}}} \\ 
  12 & 70 & 50 & 0.890 {\tiny(0.004)} & 0.888 {\tiny(0.004)} & {\textbf{0.892 {\tiny(0.004)}}} & 0.892 {\tiny(0.004)} \\ 
  12 & 70 & 70 & 0.818 {\tiny(0.003)} & 0.811 {\tiny(0.005)} & 0.818 {\tiny(0.004)} & {\textbf{0.879 {\tiny(0.003)}}} \\ 
  12 & 70 & 90 & 0.665 {\tiny(0.009)} & 0.644 {\tiny(0.008)} & 0.663 {\tiny(0.009)} & {\textbf{0.849 {\tiny(0.004)}}} \\ 
  24 & 30 & 50 & 0.930 {\tiny(0.003)} & 0.926 {\tiny(0.004)} & 0.932 {\tiny(0.004)} & {\textbf{0.932 {\tiny(0.003)}}} \\ 
  24 & 30 & 70 & 0.922 {\tiny(0.003)} & 0.916 {\tiny(0.003)} & 0.924 {\tiny(0.004)} & {\textbf{0.932 {\tiny(0.004)}}} \\ 
  24 & 30 & 90 & 0.894 {\tiny(0.004)} & 0.885 {\tiny(0.003)} & 0.896 {\tiny(0.004)} & {\textbf{0.934 {\tiny(0.004)}}} \\ 
  24 & 70 & 50 & 0.930 {\tiny(0.002)} & 0.929 {\tiny(0.002)} & 0.933 {\tiny(0.002)} & {\textbf{0.934 {\tiny(0.002)}}} \\ 
  24 & 70 & 70 & 0.863 {\tiny(0.006)} & 0.850 {\tiny(0.005)} & 0.865 {\tiny(0.007)} & {\textbf{0.932 {\tiny(0.003)}}} \\ 
  24 & 70 & 90 & 0.669 {\tiny(0.007)} & 0.643 {\tiny(0.006)} & 0.670 {\tiny(0.007)} & {\textbf{0.926 {\tiny(0.002)}}} \\ 
   \hline
\end{tabular}
\end{table}

\begin{table}[htb]
\caption{Comparisons of edgeR, DESeq2, limma and rSeqRobust using negative-binomially distributed data. See Table~\ref{tab:LN} for detailed description of the columns.} \label{tab:NB}
\centering
\begin{tabular}{ccccccc}
  \hline
$n$ & DE\% & Up\% & edgeR-robust & DESeq2 & limma-voom & rSeqRobust \\ 
  \hline
4 & 30 & 50 & {\textbf{0.835 {\tiny(0.005)}}} & 0.827 {\tiny(0.005)} & 0.825 {\tiny(0.006)} & 0.830 {\tiny(0.006)} \\ 
  4 & 30 & 70 & {\textbf{0.819 {\tiny(0.005)}}} & 0.810 {\tiny(0.006)} & 0.806 {\tiny(0.004)} & 0.804 {\tiny(0.005)} \\ 
  4 & 30 & 90 & {\textbf{0.812 {\tiny(0.006)}}} & 0.799 {\tiny(0.005)} & 0.803 {\tiny(0.006)} & 0.772 {\tiny(0.005)} \\ 
  4 & 70 & 50 & {\textbf{0.829 {\tiny(0.005)}}} & 0.820 {\tiny(0.006)} & 0.817 {\tiny(0.005)} & 0.821 {\tiny(0.005)} \\ 
  4 & 70 & 70 & {\textbf{0.790 {\tiny(0.005)}}} & 0.779 {\tiny(0.005)} & 0.780 {\tiny(0.005)} & 0.774 {\tiny(0.004)} \\ 
  4 & 70 & 90 & {\textbf{0.670 {\tiny(0.006)}}} & 0.653 {\tiny(0.005)} & 0.658 {\tiny(0.007)} & 0.625 {\tiny(0.006)} \\ 
  6 & 30 & 50 & {\textbf{0.854 {\tiny(0.006)}}} & 0.850 {\tiny(0.006)} & 0.847 {\tiny(0.007)} & 0.853 {\tiny(0.007)} \\ 
  6 & 30 & 70 & {\textbf{0.852 {\tiny(0.005)}}} & 0.844 {\tiny(0.004)} & 0.845 {\tiny(0.005)} & 0.849 {\tiny(0.005)} \\ 
  6 & 30 & 90 & {\textbf{0.832 {\tiny(0.005)}}} & 0.826 {\tiny(0.005)} & 0.824 {\tiny(0.005)} & 0.826 {\tiny(0.006)} \\ 
  6 & 70 & 50 & {\textbf{0.864 {\tiny(0.004)}}} & 0.860 {\tiny(0.004)} & 0.856 {\tiny(0.004)} & 0.863 {\tiny(0.004)} \\ 
  6 & 70 & 70 & {\textbf{0.817 {\tiny(0.004)}}} & 0.806 {\tiny(0.004)} & 0.806 {\tiny(0.005)} & 0.799 {\tiny(0.005)} \\ 
  6 & 70 & 90 & {\textbf{0.674 {\tiny(0.006)}}} & 0.651 {\tiny(0.003)} & 0.662 {\tiny(0.005)} & 0.653 {\tiny(0.006)} \\ 
  8 & 30 & 50 & {\textbf{0.887 {\tiny(0.004)}}} & 0.886 {\tiny(0.004)} & 0.880 {\tiny(0.004)} & 0.885 {\tiny(0.004)} \\ 
  8 & 30 & 70 & 0.873 {\tiny(0.003)} & 0.869 {\tiny(0.002)} & 0.867 {\tiny(0.003)} & {\textbf{0.876 {\tiny(0.003)}}} \\ 
  8 & 30 & 90 & 0.853 {\tiny(0.004)} & 0.842 {\tiny(0.004)} & 0.844 {\tiny(0.004)} & {\textbf{0.863 {\tiny(0.004)}}} \\ 
  8 & 70 & 50 & {\textbf{0.875 {\tiny(0.004)}}} & 0.872 {\tiny(0.003)} & 0.868 {\tiny(0.005)} & 0.871 {\tiny(0.004)} \\ 
  8 & 70 & 70 & 0.822 {\tiny(0.004)} & 0.802 {\tiny(0.007)} & 0.813 {\tiny(0.005)} & {\textbf{0.837 {\tiny(0.005)}}} \\ 
  8 & 70 & 90 & 0.686 {\tiny(0.008)} & 0.660 {\tiny(0.007)} & 0.674 {\tiny(0.008)} & {\textbf{0.764 {\tiny(0.005)}}} \\ 
  12 & 30 & 50 & 0.890 {\tiny(0.003)} & {\textbf{0.891 {\tiny(0.004)}}} & 0.885 {\tiny(0.004)} & 0.890 {\tiny(0.004)} \\ 
  12 & 30 & 70 & 0.896 {\tiny(0.003)} & 0.895 {\tiny(0.004)} & 0.889 {\tiny(0.003)} & {\textbf{0.901 {\tiny(0.003)}}} \\ 
  12 & 30 & 90 & 0.876 {\tiny(0.003)} & 0.863 {\tiny(0.004)} & 0.866 {\tiny(0.004)} & {\textbf{0.898 {\tiny(0.004)}}} \\ 
  12 & 70 & 50 & 0.897 {\tiny(0.003)} & 0.898 {\tiny(0.003)} & 0.891 {\tiny(0.003)} & {\textbf{0.899 {\tiny(0.003)}}} \\ 
  12 & 70 & 70 & 0.841 {\tiny(0.005)} & 0.828 {\tiny(0.006)} & 0.832 {\tiny(0.005)} & {\textbf{0.891 {\tiny(0.003)}}} \\ 
  12 & 70 & 90 & 0.668 {\tiny(0.006)} & 0.637 {\tiny(0.005)} & 0.651 {\tiny(0.006)} & {\textbf{0.864 {\tiny(0.003)}}} \\ 
  24 & 30 & 50 & {\textbf{0.929 {\tiny(0.003)}}} & 0.928 {\tiny(0.003)} & 0.923 {\tiny(0.003)} & 0.928 {\tiny(0.003)} \\ 
  24 & 30 & 70 & 0.914 {\tiny(0.004)} & 0.915 {\tiny(0.004)} & 0.912 {\tiny(0.003)} & {\textbf{0.923 {\tiny(0.003)}}} \\ 
  24 & 30 & 90 & 0.889 {\tiny(0.003)} & 0.878 {\tiny(0.004)} & 0.881 {\tiny(0.004)} & {\textbf{0.921 {\tiny(0.003)}}} \\ 
  24 & 70 & 50 & 0.926 {\tiny(0.002)} & {\textbf{0.928 {\tiny(0.002)}}} & 0.921 {\tiny(0.002)} & 0.925 {\tiny(0.002)} \\ 
  24 & 70 & 70 & 0.854 {\tiny(0.005)} & 0.834 {\tiny(0.007)} & 0.841 {\tiny(0.005)} & {\textbf{0.923 {\tiny(0.003)}}} \\ 
  24 & 70 & 90 & 0.666 {\tiny(0.007)} & 0.630 {\tiny(0.004)} & 0.648 {\tiny(0.007)} & {\textbf{0.925 {\tiny(0.002)}}} \\ 
   \hline
\end{tabular}
\end{table}

\subsection{Comparison with existing methods using real RNA-Seq data}

Again, we compare our method with all three other methods using real RNA-Seq data generated from the SEQC project~\citep{seqc2014comprehensive}. We use the data table \textit{ILM\_refseq\_gene\_AGR} from the \textbf{seqc} R package for our experiment. There are four biological samples A, B, C and D, where A and B are two different human RNA reference libraries (Aglient's Universal Human Reference RNA and Life Technologies's Human Brain Reference RNA), and C and D are different mixtures of A and B with C=$75\%$A+$25\%$B and D=$25\%$A+$75\%$B, respective. Each of the four biological samples were sequenced with four replicates so that there are a total of 16 RNA-Seq samples. As a validation, the TaqMan RT-PCR technology was also used to measure 955 genes selected from all 25,794 genes for all 16 samples~\citep{seqc2014comprehensive}. 

We compare four methods in three comparisons (A vs B, A vs C and A vs D), each one is a two-group comparison with four replicates in each group (i.e., $n=8$). Since A and B are from rather different human RNA reference libraries, we expect many DE genes for the comparison of A vs B, less DE genes for A vs D, and even less DE genes for A vs C. First, using the TaqMan RT-PCR measurements as the gold standard, we define DE genes to be those with $p$-value from two sample $t$-test (i.e., 4 vs 4 samples in each group based on the RT-PCR data) smaller than 0.05, and fold change greater than 2, as these are commonly used criteria in DE studies. With such criteria, we identify 54.2\%, 15.3\% and 47.9\% of all 955 genes as DE genes, respectively, for the three comparisons, which is consistent with our expectation. Then, we run the four methods on the RNA-Seq data from all 25,794 genes for each of the three comparisons and extract the $p$-values for the 955 genes with TaqMan data. Finally, using the DE and non-DE genes identified in the total 955 genes with TaqMan measurements, we compare the four methods in terms of their ROC curves and AUCs. The results are given in Table~\ref{tab:auc.real.data} and Figure~\ref{fig:roc.real.data}. We can see that rSeqRobust and limma-voom are the two best performing methods, with rSeqRobust slightly outperforming limma-voom. We also notice that edgeR-robust and DESeq2 both report $p$-values of exactly zero for many genes, which dampened their performance. 

We also report the running times (in seconds) of all four methods for each of the three comparisons in Table~\ref{tab:time.real.data}. We can see that rSeqRobust is faster than all other methods.

\begin{table}[htb]
\caption{Comparisons of edgeR, DESeq2, limma and rSeqRobust using real RNA-Seq data from the SEQC project~\citep{seqc2014comprehensive}. The table shows percent of DE genes (DE\%), percent of up-regulated genes among all the DE genes (Up\%), as well as AUCs for all four methods in three comparisons (A vs B, A vs C and A vs D). } \label{tab:auc.real.data}
\centering
\begin{tabular}{ccccccc}
  \hline
 & DE\% & Up\% & edgeR-robust & DESeq2 & limma-voom & rSeqRobust \\ 
  \hline
A vs B & 54.2 & 56.8 & 0.670 & 0.696 & {\textbf{0.883}} & 0.882 \\ 
  A vs C & 15.3 & 0.7 & 0.927 & 0.906 & 0.956 & {\textbf{0.972}} \\ 
  A vs D & 47.9 & 56.7 & 0.742 & 0.742 & 0.868 & {\textbf{0.903}} \\ 
   \hline
\end{tabular}
\end{table}

\begin{table}[htb]
\caption{Comparisons of of edgeR, DESeq2, limma and rSeqRobust using real RNA-Seq data from the SEQC project~\citep{seqc2014comprehensive}. The table show the running times (in seconds) for all four methods in three comparisons (A vs B, A vs C and A vs D).} \label{tab:time.real.data}
\centering
\begin{tabular}{ccccc}
  \hline
 & edgeR-robust & DESeq2 & limma-voom & rSeqRobust \\ 
  \hline
A vs B & 56.230 & 16.140 &  3.290 & {\textbf{ 1.710}} \\ 
  A vs C & 53.260 & 15.750 &  3.360 & {\textbf{ 1.720}} \\ 
  A vs D & 53.170 & 15.470 &  3.330 & {\textbf{ 1.700}} \\ 
   \hline
\end{tabular}
\end{table}

\begin{figure}[htb]
\begin{center}
\includegraphics[scale=0.4]{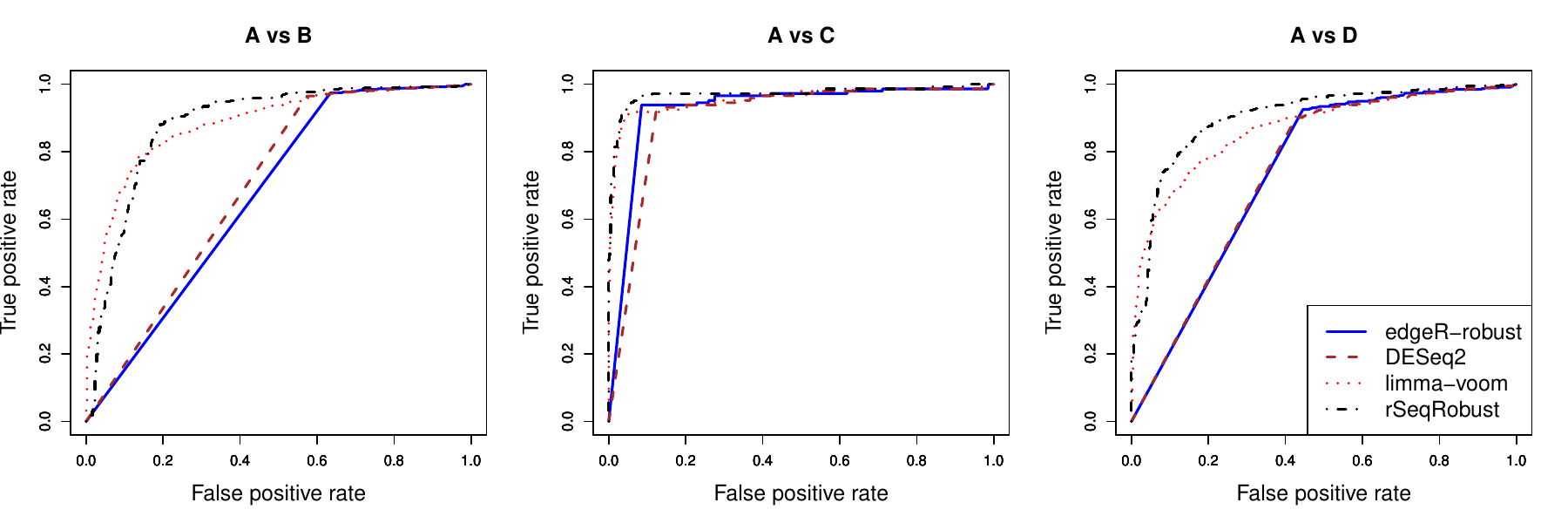}\\
\end{center}
\caption{Comparisons of edgeR, DESeq2, limma and rSeqRobust using real RNA-Seq data from the SEQC project~\citep{seqc2014comprehensive}. The figures show the ROC curves for all four methods in three comparisons (A vs B, A vs C and A vs D). \label{fig:roc.real.data}}
\end{figure}

\section{Discussion}
\label{sec:discussion}

It is shown in our simulations that our proposed approach is able to reliably normalize the data and detect differential expression in some cases when more than half of the genes are differentially expressed in an asymmetric manner. This is hard to achieve with other robust methods such as the median or trimmed mean based normalization approaches~\citep{Anders2010,Robinson2010a}. This is attributed to the L0-penalized likelihood used by our model. Typically, L0-penalized models are difficult to fit due to their non-convexity. However, in our case it is easily manageable once we reduce the model fitting to a univariate or bivariate optimization problem. Adding an L0 penalty $p(\gamma)$ results in hard thresholding on $\gamma$, which has been shown to be a general case in~\citet{She2011}. The hard thresholding facilitates the inference on indicators $\{\tau_i\}_{i=1}^m$ for differential gene expression.

In the past decade, many methods and software packages have been developed for differential expression analysis from RNA-Seq data. In this paper, we compare our method with edgeR-robust, DESeq2 and limma-voom since they have excellent performance according to several studies focusing on thorough comparisons of methods and software packages for DE detection~\citep{Rapaport2013,soneson2013comparison,zhang2014comparative,seyednasrollah2015comparison}. Also, all three methods provide user-friendly R packages for convenient comparison based on simulated and real data. Nevertheless, there are still many other excellent and widely used software packages available for the users to choose from, such as cuffdiff~\citep{trapnell2012differential} and EBSeq~\citep{leng2013ebseq}. This paper is not intended to be comprehensive, but rather to present a new way to integrate normalization and DE gene detection.

The R program for generating the results in this paper is available at \url{http://www-personal.umich.edu/~jianghui/rseqrobust/}.

\section*{Appendix}
\label{sec:appendix}

\begin{proof}[Proof of the unit-free property]
Let $c_{sij}, s=1,\ldots,S, i=1,\ldots,m, j=1,\ldots,n_s$ be the read count for the $i$-th gene in the $j$-th sample in the $s$-th group. Assume that a small positive number $\epsilon>0$ has been added to all the read counts so that $c_{sij}\geq\epsilon>0, \forall s, i, j$. Let $cpm_{sij}$, $rpkm_{sij}$ and $tpm_{sij}$ be the corresponding CPM, RPKM and TPM values for the $i$-th gene in the $j$-th sample in the $s$-th group defined as follows, respectively.
$$
\begin{array}{lll}
cpm_{sij}&=& 10^6c_{sij}/\sum_ic_{sij}\\
rpkm_{sij}&=& 10^3cpm_{sij}/l_i\\
tpm_{sij}&=& 10^6rpkm_{sij}/\sum_irpkm_{sij}
\end{array}
$$
We have 
$$\begin{array}{lll}
\log(cpm_{sij})&=& \log(c_{sij}) + A_{sj}\\
\log(rpkm_{sij})&=& \log(cpm_{sij}) + B_i = \log(c_{sij}) + A_{sj} + B_i\\
\log(tpm_{sij})&=& \log(rpkm_{sij}) + C_{sj} = \log(c_{sij}) + A_{sj} + B_i + C_{sj}
\end{array}
$$
\end{proof}
where $A_{sj}=\log(10^6/\sum_ic_{sij})$, $B_i=\log(10^3/l_i)$ and $C_{sj}=\log(10^6/\sum_irpkm_{sij})$ are constants.
Due to the structure of model~(\ref{model}), when we change $x_{sij}$ from $\log(c_{sij})$ to one of $\log(cpm_{sij})$, $\log(rpkm_{sij})$ and $\log(tpm_{sij})$, the constants $A_{sj}$ and $C_{sj}$ will be absorbed into the parameter $d_{sj}$, and the constant $B_i$ will be absorbed into the parameters $\{\mu_{si}\}_{s=1}^S$ simultaneously. As a result, $\tau_i$, the parameter of interest, will not change because it is determined only by the differences among $\{\mu_{si}\}_{s=1}^S$.

\begin{proof}[Proof of Proposition~\ref{solution}]
We have 
$$\begin{array}{ll}
f(\mu,\gamma,d)=&\displaystyle\sum_{i=1}^m\frac1{2\sigma_i^2}\left(\sum_{j=1}^{n_1}(x_{1ij}-\mu_i-d_{1j})^2+\sum_{s=2}^S\sum_{j=1}^{n_s}(x_{sij}-\mu_i-\gamma_{si}-d_{sj})^2\right)\\
&\displaystyle+\sum_{i=1}^m\alpha_i1(\sum_{s=2}^S|\gamma_{si}|>0)
\end{array}$$
Therefore
$$\frac{\partial f}{\partial d_{1j}}=\sum_{i=1}^{m}-\frac1{\sigma_i^2}(x_{1ij}-\mu_i-d_{1j})=0$$
which gives
$$d_{1j}=\frac{\sum_{i=1}^{m}\frac1{\sigma_i^2}(x_{1ij}-\mu_i)}{\sum_{i=1}^{m}\frac1{\sigma_i^2}}$$
Consequently
$$d_{1j}-d_{11}=\frac{\sum_{i=1}^{m}\frac1{\sigma_i^2}(x_{1ij}-x_{1i1})}{\sum_{i=1}^{m}\frac1{\sigma_i^2}}=d^\prime_{1j}$$
Since we fix $d_{11}=0$, we have $d_{1j}=d^\prime_{1j}$.
Similarly, we have 
$$d_{sj}-d_{s1}=\frac{\sum_{i=1}^{m}\frac1{\sigma_i^2}(x_{sij}-x_{si1})}{\sum_{i=1}^{m}\frac1{\sigma_i^2}}=d^\prime_{sj}, s=2,\ldots,S$$
Denote $d_{s1}$ as $d_s$, $s=1,\ldots,S$, we have $d_1=d_{11}=0$ and $d_{sj}=d_s+d^\prime_{sj}, s=1,\ldots,S$.
Now we have
$$\begin{array}{ll}
f(\mu,\gamma,d)=&\displaystyle\sum_{i=1}^m\left(\frac1{2\sigma_i^2}\sum_{j=1}^{n_1}(x_{1ij}-\mu_i-d^\prime_{1j})^2+\frac1{2\sigma_i^2}\sum_{s=2}^S\sum_{j=1}^{n_s}(x_{sij}-\mu_i-\gamma_{si}-d_s-d^\prime_{sj})^2\right.\\
&\displaystyle\left.+\alpha_i1(\sum_{s=2}^S|\gamma_{si}|>0)\right)
\end{array}$$
which can be written as
$$f(\mu,\gamma,d)=\sum_{i=1}^mh_i(\mu_i, \gamma_{2i},\ldots,\gamma_{Si})$$
where the $\mu_i$'s and the $\gamma_{si}$'s can all be considered as functions of $(d_2,\ldots,d_S)$. Here we first work out these functions, i.e., solutions for the $\mu_i$'s and the $\gamma_{si}$'s when the $d_s$'s are considered fixed. Fixing $d_2,\ldots,d_S$, $f$ can be minimized by minimizing each $h_i$ separately. When $\sum_{s=2}^S|\gamma_{si}|>0$, $h_i$ is easily minimized by 
$$\mu_i = \frac1{n_1} \sum_{j=1}^{n_1}(x_{1ij}-d^\prime_{1j})=\mu^\prime_{1i}$$
and 
$$\gamma_{si}=\frac1{n_s} \sum_{j=1}^{n_s}(x_{sij}-\mu^\prime_{1i}-d_s-d^\prime_{sj})=\frac1{n_s} \sum_{j=1}^{n_s}(x_{sij}-d^\prime_{sj})-\mu^\prime_{1i}-d_s=\mu^\prime_{si} -\mu^\prime_{1i}-d_s$$
Similarly, when $\sum_{s=2}^S|\gamma_{si}|=0$, i.e., $\gamma_{2i}=\cdots=\gamma_{Si}=0$, we have 
$$\mu_i=\frac1{\sum_{s=1}^Sn_s}(n_1\mu^\prime_{1i}+\sum_{s=2}^Sn_s(\mu^\prime_{si}-d_s))=\frac{\sum_{s=1}^Sn_s(\mu^\prime_{si}-d_s)}{n}$$
Changing from $\sum_{s=2}^S|\gamma_{si}|>0$ to $\sum_{s=2}^S|\gamma_{si}|=0$, the increase in the maximum achievable value of $h_i$ is
$$
\begin{array}{l}
\displaystyle h_i\left(\frac{\sum_{s=1}^Sn_s(\mu^\prime_{si}-d_s)}{n}, 0,\dots,0\right)-h_i(\mu^\prime_{1i}, \mu^\prime_{2i} -\mu^\prime_{1i}-d_2,\dots,\mu^\prime_{Si} -\mu^\prime_{1i}-d_S)\\
\displaystyle=\frac1{2\sigma_i^2}\left\{\sum_{s=1}^Sn_s(\mu_{si}^\prime-d_s)^2-\frac1n\left[\sum_{s=1}^S(n_s(\mu_{si}^\prime-d_s))\right]^2\right\}-\alpha_i
\end{array}
$$
which is a quadratic function of $d_2,\ldots,d_S$. Therefore
$$
	\gamma_{si} = \left\{
	\begin{array}{ll}
		0  &\mbox{ if } \displaystyle\frac1{2\sigma_i^2}\left\{\sum_{s=1}^Sn_s(\mu_{si}^\prime-d_s)^2-\frac1n\left[\sum_{s=1}^S(n_s(\mu_{si}^\prime-d_s))\right]^2\right\}<\alpha_i\\
		\mu^\prime_{si}-\mu^\prime_{1i}-d_s &\mbox{ otherswise}		
	\end{array}\right.
$$
Now we only need to solve for $d_2,\ldots,d_S$. We have
$$
\begin{array}{l}
\displaystyle d_2,\ldots,d_S=\argmin_{d_2,\ldots,d_S}\sum_{i=1}^m\min\left[h_i\left(\frac{\sum_{s=1}^Sn_s(\mu^\prime_{si}-d_s)}{n}, 0,\dots,0\right)\right.\\
,h_i(\mu^\prime_{1i}, \mu^\prime_{2i} -\mu^\prime_{1i}-d_2,\dots,\mu^\prime_{Si} -\mu^\prime_{1i}-d_S)]
\end{array}
$$
which can be simplified as
$$d_2,\ldots,d_S=\argmin_{d_2,\ldots,d_S}\sum_{i=1}^m\min\left(\frac1{2\sigma_i^2}\left\{\sum_{s=1}^Sn_s(\mu_{si}^\prime-d_s)^2-\frac1n\left[\sum_{s=1}^S(n_s(\mu_{si}^\prime-d_s))\right]^2\right\}-\alpha_i, 0\right)$$
since $h_i(\mu^\prime_{1i}, \mu^\prime_{2i} -\mu^\prime_{1i}-d_2,\dots,\mu^\prime_{Si} -\mu^\prime_{1i}-d_S)$ is actually not a function of $d_2,\ldots,d_S$.
Therefore,
$$d_2,\ldots,d_S=\argmin_{d_2,\ldots,d_S}\sum_{i=1}^m\min\left(\frac1{2\sigma_i^2}\left\{\sum_{s=1}^Sn_s(\mu_{si}^\prime-d_s)^2-\frac1n\left[\sum_{s=1}^S(n_s(\mu_{si}^\prime-d_s))\right]^2\right\}, \alpha_i\right)$$ 
\end{proof}

\bibliographystyle{abbrvnat}
\bibliography{manuscript}
\end{document}